\documentclass[a4paper]{beatcs} %final option removed for adding page numbers
\usepackage[dvipsnames,svgnames, table]{xcolor}
\usepackage{enumerate}

\usepackage{todonotes}
\hypersetup{hidelinks,breaklinks}

\usepackage{pifont}
\usepackage{bbding}

\newtheorem{definition}{Definition}

\newtheorem{proposition}[definition]{Proposition}
\newtheorem{lemma}[definition]{Lemma}
\newtheorem{corollary}[definition]{Corollary}

\newtheorem{example}[definition]{Example}

\usepackage{mathcommand}
\usepackage[notion, quotation, electronic]{knowledge}

\RequirePackage{refcount}
\RequirePackage{enumerate}
\RequirePackage{enumitem}

\usepackage{colours}
\usepackage{hyperref}
\IfKnowledgePaperModeTF{
	}{
		\knowledgestyle{intro notion}{color={Dark Ruby Red}, emphasize}
		\knowledgestyle{notion}{color={Dark Blue Sapphire}}
		\hypersetup{
				colorlinks=true,
				breaklinks=true,
				linkcolor={}, % Links to sections, pages, etc.
				citecolor={}, % Links to bibliography
				filecolor={Blue Marine}, % Links to local file
				urlcolor={Blue Marine},
			}
	}

\knowledge{notion}
| DFA
| DFAs
| state-based DFA

\knowledge{notion}
| acceptance condition of a transition-based DFA
| transition-based DFAs
| transition-based DFA

\knowledge{notion}
 | run
 | runs

\knowledge{notion}
 | state-based acceptance condition
 | state-based
 | state-based automaton
 | state-based acceptance

\knowledge{notion}
 | transition-based acceptance condition
 | transition-based acceptance
 | transition-based
 | transition-based automaton
 | acceptance condition
 | acceptance conditions
 | condition
 | Transition-based

\knowledge{notion}
 | prefix-independent
 | prefix-independence

\knowledge{notion}
 | accepting@states
 | accepting state
 | accepting@state
 | accepting B\"uchi states

\knowledge{notion}
 | B\"uchi
 | B\"uchi automata
 | \Buchi

\knowledge{notion}
 | Muller condition
 | Muller
 | Muller language
 | \Muller

\knowledge{notion}
 | generalised B\"uchi
 | generalised coB\"uchi
 | \genBuchi

\knowledge{notion}
 | history-deterministic
 | History-deterministic
 | HD

\knowledge{notion}
 | coB\"uchi
 | \coBuchi

\knowledge{notion}
 | games on graphs
 | game
 | games

\knowledge{notion}
 | locally bijective morphisms
 | locally bijective morphism

\knowledge{notion}
 | minimisation problem
 | minimisation

\knowledge{notion}
 | parity
 | \parity

\knowledge{notion}
 | automaton
 | automata
 | $\omega $-automata
 | $\oo$-automaton
 | $\omega$-automaton
 | $\oo$-automata

\knowledge{notion}
 | Rabin

\knowledge{notion}
 | winning condition

\knowledge{notion}
 | positional
 | positionality
 | positional@lang
 | positional languages

\knowledge{notion}
 | positional@strat
 | positional strategies
 | positional strategy

\knowledge{notion}
 | finite-memory

\knowledge{notion}
 | winning strategy

\knowledge{notion}
 | bipositional
 | bipositionality

\knowledge{notion}
 | monotone
 | monotonicity
 | monotone graph

\knowledge{notion}
 | memory

\knowledge{notion}
 | $\ee $-completable
 | $\ee $-completion

\knowledge{notion}
 | NFAs
 | NFA

\mathchardef\hyphen=45 %Decimal
\usepackage[super]{nth}

\newcommand\restr[2]{{% we make the whole thing an ordinary symbol
		\left.\kern-\nulldelimiterspace % automatically resize the bar with \right
		#1 % the function
		\littletaller % pretend it's a little taller at normal size
		\right|_{#2} % this is the delimiter
}}

\newcommand{\littletaller}{\mathchoice{\vphantom{\big|}}{}{}{}}

\newrobustcmd\inv[1]{#1^{-1}}

\newcommand{\tand}{\text{ and }}

\newcommand{\NP}{\ensuremath{\mathsf{NP}}}
\newcommand{\NPc}{\ensuremath{\mathsf{NP}\hyphen\mathrm{complete}}}

\newcommand{\PSPACE}{\ensuremath{\mathsf{PSPACE}}}

\DeclareMathAlphabet{\mathpzc}{OT1}{pzc}{m}{it}

\newrobustcmd{\NN}{\mathbb{N}}
\newrobustcmd{\ZZ}{\mathbb{Z}}
\newrobustcmd{\QQ}{\mathbb{Q}}
\newrobustcmd{\RR}{\mathbb{R}}
\newrobustcmd{\CC}{\mathbb{C}}
\newrobustcmd{\WW}{\mathbb{W}}

\newrobustcmd{\I}{\mathcal{I}}
\newrobustcmd{\F}{\mathcal{F}}
\newrobustcmd{\D}{\mathcal{D}}
\newrobustcmd{\N}{\mathcal{N}}
\newrobustcmd{\G}{\mathcal{G}}

\newrobustcmd{\M}{\mathcal{M}}
\newrobustcmd{\Q}{\mathcal{Q}}
\newrobustcmd{\C}{\mathcal{C}}
\newrobustcmd{\A}{\mathcal{A}}
\newrobustcmd{\B}{\mathcal{B}}
\newrobustcmd{\Z}{\mathcal{Z}}
\newrobustcmd{\R}{\mathcal{R}}
\newrobustcmd{\T}{\mathcal{T}}
\newrobustcmd{\U}{\mathcal{U}}
\newrobustcmd{\W}{\mathcal{W}}

\renewcommand{\P}{\mathcal{P}}

\newrobustcmd{\kk}{\kappa}
\newrobustcmd{\uu}{\upsilon}
\newrobustcmd{\dd}{\delta}

\renewcommand{\ss}{\sigma}

\newrobustcmd{\rr}{\rho}

\newrobustcmd{\bb}{\beta}

\newrobustcmd{\oo}{\omega}
\newrobustcmd{\pp}{\varphi}

\newrobustcmd{\ee}{\varepsilon}
\renewcommand{\epsilon}{\varepsilon}

\renewcommand{\SS}{\Sigma}
\newrobustcmd{\GG}{\Gamma}
\newrobustcmd{\DD}{\Delta}

\knowledgerenewmathcommand\nu{\cmdkl{\LaTeXnu}}
\knowledgenewmathcommand\nuAcd{\cmdkl{\LaTeXnu}}
\knowledgerenewmathcommand\eta{\cmdkl{\LaTeXeta}}
\newrobustcmd{\SOneS}{\ensuremath{\mathrm{S1S}}}
\newrobustcmd{\LTL}{\kl[\LTL]{\ensuremath{\mathrm{LTL}}}}
\knowledge{\LTL}[\LTLf]{notion}
\newrobustcmd{\LTLf}{\kl[\LTLf]{\ensuremath{\mathrm{LTL}_f}}}

\newrobustcmd{\infOften}{\kl[\infOften]{\mathtt{Inf}}}
\knowledge{\infOften}{notion}
\newrobustcmd{\finOften}{\kl[\finOften]{\mathtt{Fin}}}
\knowledge{\finOften}{notion}
\newrobustcmd{\noOcc}{\kl[\noOcc]{\mathtt{No}}}

\newcommand\iffdef{\mathrel{\overset{\makebox[0pt]{\mbox{\normalfont\footnotesize def}}}{\iff}}}

\newcommand{\re}[1]{\xrightarrow{#1}}
\usetikzlibrary{calc,decorations.pathmorphing,shapes} %%

\newcounter{sarrow}

\newrobustcmd{\parity}{\kl[\parity]{\mathsf{parity}}}
\newrobustcmd{\Buchi}{\kl[\Buchi]{\mathsf{Buchi}}}
\newrobustcmd{\coBuchi}{\kl[\coBuchi]{\mathsf{coBuchi}}}
\newrobustcmd{\Muller}{\kl[\Muller]{\mathsf{Muller}}}
\newrobustcmd{\genBuchi}{\kl[\genBuchi]{\mathsf{genBuchi}}}
\newcommand{\Eve}{\mathrm{Eve}}
\newcommand{\Adam}{\mathrm{Adam}}

\newrobustcmd{\VEve}{\kl[\VEve]{V_{\mathrm{Eve}}}}
\knowledge\VEve[V_P|\widetilde {V}_\Eve|\tilde {V}_\Eve ]{notion}
\newrobustcmd{\VAdam}{\kl[\VAdam]{V_{\mathrm{Adam}}}}
\knowledge\VAdam[\widetilde {V}_\Adam ]{notion}

\newrobustcmd{\strat}{\mathsf{strat}}

\newrobustcmd{\Verts}[1]{\kl[\Verts]{V(#1)}}
\knowledge{\Verts}{notion}
\newrobustcmd{\Edges}[1]{\kl[\Edges]{E(#1)}}
\knowledge{\Edges}{notion}

\newcommand{\init}{{\mathsf{init}}}

\newcommand{\acc}{\mathit{Acc}}
\newcommand{\win}{\mathit{Win}}

\newcommand{\tr}{\mathrm{tr}}
\newcommand{\st}{\mathrm{st}}

\newrobustcmd{\accLast}{\kl[\accLast]{\acc_\mathrm{Last}}}
\knowledge{\accLast}{notion}

\newrobustcmd{\accReach}{\kl[\accReach]{\acc_{\mathrm{Reach}}}}
\knowledge{\accReach}{notion}

\newrobustcmd{\sizeDet}{\kl[\sizeDet]{\mathsf{sizeDet}}}
\knowledge\sizeDet{notion}

\newrobustcmd{\eqeps}{\mathrel{\kl[\eqeps]{\dot{\sim}_L}}}
\knowledge\eqeps{notion}

\NewDocumentCommand{\pbVertexCover}{}{\normalfont{\textsc{\small{Vertex Cover}}}}
\title{Transition-based vs stated-based acceptance for automata over infinite words}

\author{Antonio Casares\thanks{University of Warsaw, Poland. Email: \url{antoniocasaressantos@gmail.com}\\ This work was supported by the Polish National Science Centre (NCN) grant ``Polynomial finite state computation'' (2022/46/A/ST6/00072).}}

\RequirePackage[
    backend=biber,
    style=alphabetic,   % or any biblatex style like numeric, authoryear, etc.
    sorting=nty,
    maxbibnames=99,
    doi=true,
    url=true,
    isbn=false
  ]{biblatex}

\DeclareFieldFormat{titlecase}{\MakeSentenceCase*{#1}} %%Forces lower-case everywhere in references, except for first letter
\begin{document}

\maketitle

% \pagestyle{plain} %%Add line numbers
% \pagenumbering{arabic}

\begin{abstract}
 Automata over infinite objects are a well-established model with applications in logic and formal verification.
 Traditionally, acceptance in such automata is defined based on the set of states visited infinitely often during a run. 
 However, there is a growing trend towards defining acceptance based on transitions rather than states.

 In this survey, we analyse the reasons for this shift and advocate using transition-based acceptance in the context of automata over infinite words. 
 We present a collection of problems where the choice of formalism has a major impact and discuss the causes of these differences. 
\end{abstract}

%
%
%\paragraph*{}
%This document contains hyperlinks. \AP Each occurrence of a "notion" is linked to its ""definition"". On an electronic device, the reader can click on words or symbols (or just hover over them on some PDF readers) to see their definition.

\setcounter{tocdepth}{1}%%Only Sections in TOC
\tableofcontents

%  \subparagraph{ {\small Acknowledgements.}}
%  {\small I warmly thank Thomas Colcombet for many discussions on the subject, Alexandre Duret-Lutz for valuable comments on transition-based DFAs and for sharing many historical references, 
% Géraud Sénizergues for pointing me to the works of Bertrand Le Saëc and 
% Pierre Ohlmann for helpful feedback on a draft of this paper.}

\newpage

\section{Introduction}\label{sec:intro}
Automata theory is a central and long-established topic in computer science. The definition of finite automata has barely suffered any modification since the introduction of non-deterministic automata by Rabin and Scott~\cite{RS59FiniteAutomata}.
However, the generalisation of automata to infinite words presents less stable definitions, as different modes of acceptance are best suited to different situations.
Recently, there has been a shift in the community towards using transitions instead of states to encode the acceptance condition of $\omega$-automata.
%In particular, in recent years, the community witnesses a shift of formalism by using \emph{transitions}, instead of \emph{states}, to encode the acceptance condition of $\omega$-automata.
In this survey, we analyse the reasons for this shift and advocate using "transition-based" acceptance in the context of automata over infinite words.

\subparagraph{Automata over infinite words.} 
\AP An ""automaton"" over an input alphabet $\Sigma$ is given by
\begin{itemize}
	\item a finite set of states $Q$,
	\item a set of transitions $\Delta \subseteq Q\times \Sigma \times Q$,
	\item a set of initial states $Q_\init\subseteq Q$, and
	\item an "acceptance condition".
\end{itemize}
\AP A ""run"" over a (finite or infinite) word $w$ is a path in the automaton starting in $Q_\init$ and with transitions labelled by the letters of $w$.
The "acceptance condition" is thus a representation of the set of paths that are accepting.

If the automaton works over finite words, the "acceptance condition" usually takes the form of a subset of final states: a run is accepting if it ends in one of them (see Section~\ref{sec:finite-words} for further discussions on finite words).
For automata over infinite words the situation is more complicated. Several acceptance conditions are commonly used, but they differ in expressive power and the complexity of related problems (see for instance~\cite{Boker18WhyTypes}).
The main focus of this paper is the following dichotomy: Should we use states or transitions to encode the acceptance condition of automata over infinite words? More formally, we will consider "acceptance conditions" of one of the following forms.	

\AP A ""state-based acceptance condition"" is a language $\acc\subseteq Q^\omega$.
\AP A ""transition-based acceptance condition"" is a language $\acc\subseteq \Delta^\omega$.%\footnote{To obtain a well-behaved class of automata, these languages should be "prefix-independent". See Section~\ref{sec:finite-words} for details.}
Usually, we represent them via a finite set of colours $C$, a colouring function $\gamma\colon Q\to C$ (resp. $\gamma\colon \Delta \to C$) and a language $\acc'\subseteq C^\omega$.
%In this case, we say that $(\gamma, \acc')$ is the ""acceptance condition"" of the automaton.
That is, we see automata as transducers $\Sigma^\omega \to C^\omega$, and the acceptance condition is given by a subset of the image. 
\AP Two languages that are commonly used as "acceptance conditions" are:
 
\begin{itemize}
	\item $\intro*\Buchi = \{w\in \{-, \bullet\}^\omega \mid w \text{ contains } \bullet \text{ infinitely often}\}$.  We may refer to states (resp. transitions) coloured with $\bullet$ as ""accepting@@states"".
	\item $\intro*\coBuchi = \{w\in \{-,\text{\ding{55}}\}^\omega \mid w \text{ contains } \text{\ding{55}} \text{ finitely often}\}$.
\end{itemize}

We show examples of "B\"uchi" automata in Figure~\ref{fig:Buchi-automata}. 

\begin{figure}[ht]
	\centering
	\hspace{5mm}
	\begin{minipage}[c]{0.45\textwidth} 
		\includegraphics[scale=1.1]{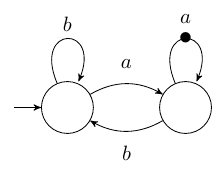}
	\end{minipage} 
	\hspace{-5mm}
	\begin{minipage}[c]{0.45\textwidth} 
		\includegraphics[scale=1.1]{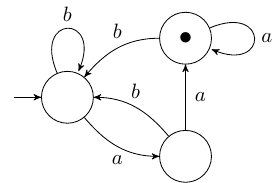}
	\end{minipage} 
	\caption{Two "B\"uchi" automata recognising the language of words containing infinitely many factors `$aa$'.
	The automaton on the left uses "transition-based acceptance", while the automaton on the right is "state-based".}
	\label{fig:Buchi-automata}
\end{figure}

\vspace{-2mm}
\subparagraph{The origins.} 
Automata over infinite words were first introduced by B\"uchi in the 60s~\cite{Buchi1962decision}, using a formalism that put the acceptance condition over states.\footnote{Corroborating this claim can be quite challenging. The use of state-based acceptance can be observed, for instance, in the first line of the proof of Lemma~12 (page 8).	In B\"uchi's 1969 paper with Landweber~\cite{BL69Definability}, this is a bit simpler to appreciate in the definitions of $\mathrm{Sup}Z$ and $U$, in the second page of the paper.}
The tradition of employing state-based acceptance persisted in all subsequent classic foundational works on $\oo$-automata: Muller's paper at the origin of the "Muller condition"~\cite{Muller1963infSequences}, Landweber's study of the complexity of $\oo$-regular languages~\cite{Landweber69OmegaAut}, McNaughton's works on $\oo$-regular expressions~\cite{McNaughton1966Testing} and infinite games~\cite{McNaughton93InfiniteGames}, Rabin's decidability result of S2S~\cite{Rabin69S2S}, Wagner's paper introducing a hierarchy of complexity~\cite{Wagner1979omega}, etc. Following this tradition, virtually all handbooks and surveys about automata on infinite objects use state-based acceptance~\cite{Eilenberg74A, Thomas90, Thomas97LangAutLog, GradelThomasWilke2002AutLogGames, PP04InfiniteWords, BK08PrinciplesMC, Kupferman2018Handbook, BloemCJ2018Handbook, WS21HandbookAut, Loding21Infinite,Esparza2023automata}. To the best of our knowledge, the only exceptions are the recent book \emph{Games on Graphs} edited by Fijalkow~\cite{bookGames25}, and the book \emph{An Automata Toolbox} by Bojańczyk~\cite{automataToolbox25}.

\subparagraph{The rise of transition-based acceptance.} 
Automata with ``effects'' on transitions, such as sequential transducers\footnote{Transducers with output on states were  also considered by Moore~\cite{Moore56}. However, the model with output on transitions popularized by Mealy~\cite{Mealy55} rapidly became the norm.}~\cite[Sect.8]{Shannon38}\cite{Mealy55,Schutzenberger61Transducers}
or weighted automata~\cite{Schutzenberger61Family} have been considered since the beginnings of automata theory.
%Trakhtenbrot~\cite[p.1005]{Trakhtenbrot58Operators}
"Transition-based" "$\omega$-automata" made their first, though modest, appearance in the mid-80s.
To the best of our knowledge, their first occurrences were in Michel's work on the connection between 
\AP \emph{Linear Temporal Logic} ($\intro*\LTL$) and automata~\cite{Michel1984}, and in Kurshan's paper on the complementation of deterministic B\"uchi automata~\cite{Kurshan1987}.\footnote{The possibility of using "transition-based" acceptance was previously suggested in~\cite[Section~8.2]{Parker81}. Some sources~\cite{Redziejowski99} mention that "transition-based" acceptance was already suggested by Redziejowski in 1972~\cite{Redziejowski72}; unfortunately we could not get access to this paper.}
In the early 90s, Le Saëc made more systematic use of this model~\cite{Saec90Saturating,SPW91Semigroups,DSL95}.
He reintroduced "transition-based" "Muller" automata under the name of table-transition automata, and characterised which languages admit a unique \emph{morphism-minimal} "Muller" automaton: those that can be recognised by a "Muller" automaton with one state per residual of the language~\cite[Cor.~5.15]{DSL95}. This characterisation no longer holds for "state-based" automata (see~Example~\ref{ex:automaton-residuals} for an illustration on how the previous property is sensitive to the placement of the "acceptance condition").
Despite the works of Le Saëc, "transition-based" automata were used only scarcely in the following years. %~\cite{CN06,Schewe2009tighter,Colcombetz2009tight}.

Some notable exceptions to the predominant use of "state-based" automata in the 2000s are given by a series of works concerning the translation of $\LTL$ formulas to automata.
%, starting with Michel's paper~\cite{Michel1984}.
In 1999, Couvreur proposed a translation using "transition-based" "generalised B\"uchi" automata~\cite{Couvreur99}.
A similar algorithm was the base for the tool \texttt{ltl2ba} by Gastin and Oddoux~\cite{GO01} (the importance of the use of "transition-based" automata in this work is discussed in~\cite{GL02Transitions}). The use of "transition-based" acceptance in this subarea was further fostered by the tool Spot~\cite{DP04Spot,Spot2.10CAV22}, influenced by Couvreur's approach.
%, and which used "transition-based" automata by default since its first version.
%Several became relatively common in this subarea,  e.g.~\cite{CDP05,BKRS12,SEJK16Limit}. 
More recently, "transition-based" automata have been adopted in the HOA format~\cite{HOA15}, and it is the primary model in other tools such as Owl~\cite{KMS18Owl} or \texttt{ltl3tela}~\cite{MBSSZ19}.
We refer to~\cite[pages 66-67]{Duret07PhD} for an overview of the use of state-based and transition-based approaches to the translation of $\LTL$ prior to 2007. %and to~\cite[Section~2.2]{Duret17HDR} for discussions on the 
%\texttt{Rabinizer 3}~\cite{KK14Rabinizer3} used them, but it is not clear how much. Rabinizer 2 the paper does not mention transition acceptance

A turning point occurred in 2019, as Abu Radi and Kupferman proved that "transition-based" "history-deterministic" "coB\"uchi" automata can be minimised in polynomial time~\cite{AK19Min}, while Schewe showed that the corresponding problem is $\NP$-complete for "state-based" automata~\cite{Schewe20MinimisingGFG}.
Since then, there is an increasing interest for "transition-based" $\oo$-automata, and, as discussed in Sections~\ref{sec:automata} and~\ref{sec:games}, many recent results rely on the use of this model.

\subparagraph{Why was the use of state-based acceptance widespread?}
We may wonder why "state-based" "automata" were the ubiquitous model for more than 50 years. 
Probably the most influential factor is that $\omega$-automata generalise automata over finite words, for which acceptance over states is a natural choice. 
Some constructions of $\oo$-automata build on automata over finite words, and for some of these, "state-based" acceptance appears naturally.

One such example is the characterisation of languages recognised by deterministic "B\"uchi" automata as limits of languages of finite words~\cite{Landweber69OmegaAut}. %(see also~\cite[Remark~4.1]{Thomas1991AutomataOI} or~\cite[Chapter~IX]{Eilenberg74A}). 
A language $L\subseteq \SS^\oo$ can be recognised by a deterministic "B\"uchi" automaton if and only for some regular language of finite words $L_{\mathsf{fin}}\subseteq \SS^*$ we have:
\[ L = \overrightarrow{L_{\mathsf{fin}}} = \{w\in \SS^\oo \mid w \text{ contains infinitely many prefixes in } L_{\mathsf{fin}}\}. \]
Building a "state-based" "B\"uchi" automaton from a deterministic automaton recognising $L_{\mathsf{fin}}$ is easy: we just need to interpret the final states of the automaton as "accepting B\"uchi states". The converse direction follows similarly.

\subparagraph{Structure of the survey.} 
We start by showing in Section~\ref{sec:st-to-trans} that we can switch between state and "transition-based" acceptance with at most a linear blow-up. However, we already notice a key difference: going from a "state-based" automaton to a "transition-based" one does not require adding any additional state, while deciding the minimal number of states required to perform the converse transformation is $\NP$-hard (Proposition~\ref{prop:from-trans-to-states-optimal-NP}). 
In Sections~\ref{sec:automata} and~\ref{sec:games}, we study problems on "$\oo$-automata" and games where the choice between "transition-based" and "state-based" acceptance may strongly affect the complexity of a given problem. %is $\NP$-complete or solvable in polynomial time.
%In Section~\ref{sec:games}, we discuss how the placement of the "acceptance condition" impacts the study of strategy complexity in "games on graphs".
In Section~\ref{sec:finite-words} we explore "transition-based acceptance" for automata over finite words.
Finally, in Section~\ref{sec:conclusion} we discuss some of the reasons causing the striking differences between the two models.

Definitions are introduced progressively as needed. 
The reader may use the hyperlinks on technical terms to quickly see their definition.

\section{From states to transitions and vice versa}\label{sec:st-to-trans}
At first sight, it could seem that there is no great difference between "state-based" or "transition-based" acceptance: we can go from one model to the other with at most a linear blow-up. However, "transition-based" "automata" are always smaller, and going from a "state-based" "automaton" to a "transition-based" one in an optimal way is $\NP$-hard, as stated in Proposition~\ref{prop:from-trans-to-states-optimal-NP}.

\begin{proposition}\label{prop:from-states-to-trans}
	Every "state-based" automaton can be relabelled with an equivalent "transition-based acceptance condition".
\end{proposition}
\begin{proof}
	Let $\acc\subseteq Q^\omega$ be the acceptance condition of the automaton, and let $\gamma\colon \Delta \to Q$ be the function assigning to each transition $(q,a,q')$ its source state $q$. Then, $(\gamma,\acc)$ is an equivalent "transition-based acceptance condition".
\end{proof}

In general, we cannot relabel in a similar manner a "transition-based automaton" to obtain an equivalent "state-based" one.
We can, however, build an equivalent "state-based automaton" paying a small blow-up on the number of states. 

\begin{proposition}\label{prop:from-trans-to-states}
	Every "transition-based" automaton admits an equivalent "state-based" automaton with at most $|Q||\Delta|+|Q_\init|$ states.
\end{proposition}
\begin{proof}
	Let $\A$ be a "transition-based" automaton with acceptance $\acc\subseteq \Delta^\omega$.
	We define the automaton having:
	\begin{itemize}
		\item States: $(Q\times \Delta) \cup Q_\init$.
		\item Transitions: For every transition $t' = q\re{a} q'$ in $\A$, we let $(q,t) \re{a} (q',t')$, and $q \re{a} (q',t')$ if $q\in Q_\init$.
		\item Initial states: $Q_\init$.
		\item Acceptance condition: We define $\gamma\colon Q \to \Delta\cup\{x\}$ by: $\gamma(q,t) = t$ and $ \gamma(q_0)=x$ if $q_0\in Q_\init$. The "acceptance condition" is given by the colouring $\gamma$ and the language $x\cdot \acc$. 
	\end{itemize}
	It is immediate to check that the obtained automaton is equivalent to $\A$.
\end{proof}

In both proofs above, the obtained automaton is not only equivalent to the original one, but there is a bijection between the runs of both.
We formalise this idea with the notion of \emph{"locally bijective morphisms"}~\cite[Def.3.3]{CCFL24FromMtoP}.

\AP Given two automata  $\A, \A'$ over the same alphabet, a ""locally bijective morphism"" is a function $\pp\colon Q\to Q'$ such that:
\begin{itemize}
	\item $\pp(Q_\init) = Q_\init'$,
	\item for all $(q,a,q')\in \Delta$, $(\pp(q),a,\pp(q'))\in \Delta'$,
	\item for all $(p,a,p')\in \Delta'$ and $q\in \pp^{-1}(p)$, there is $q'\in \pp^{-1}(p')$ such that $(q,a,q')\in \Delta$, and
	\item a "run" $\rho$ in $\A$ is accepting if and only if $\pp(\rho)$ is accepting in $\A'$.
\end{itemize}

Intuitively, if $\pp\colon \A \to \A'$ is a "locally bijective morphism", it means that $\A$ has been obtained from $\A'$ by duplicating some of its states, for instance, via a product construction. 
For example, the automaton on the right of Figure~\ref{fig:Buchi-automata} admits a "locally bijective morphism" to the automaton on its left. 

Proposition~\ref{prop:from-states-to-trans} implies that for every "state-based" automaton there is a "transition-based" automaton of same size admiting a "locally bijective morphism" to it (the automaton itself).
However, in the other direction, deciding whether there is a small "state-based" automaton admiting a "locally bijective morphism" towards a given "transition-based" automaton is hard, already for "B\"uchi automata". 

\begin{proposition}\label{prop:from-trans-to-states-optimal-NP}
	The following problem is $\NPc$:
	\begin{center}
		\hspace{-3mm}
	\begin{tabular}{rl}
				{Input:} & A "transition-based" "B\"uchi" automaton $\A_{\tr}$ and a positive integer $n$.\\[1mm]   
				{Question:} & Is there a "state-based" "B\"uchi" automaton with $n$ states admitting\\
				&  a "locally bijective morphism" to $\A_{\tr}$?  
		\end{tabular}
	\end{center}
\end{proposition}
\begin{proof}%[Proof sketch]
	To show $\NP$-hardness, we use the reduction from $\pbVertexCover$ given by Schewe to show the $\NP$-completeness of the minimisation of "state-based" deterministic "B\"uchi" automata~\cite{Schewe10MinimisingNPComplete}. 
	
	Let $G = (V,E)$ be an undirected graph. Consider the "B\"uchi" automaton $\A_G$ over the alphabet $\Sigma = V$ with states $Q_G =V$, all of them initial, and transitions $u\re{v} v$ for every $(u,v)\in E$, and for $u=v$. For the $\Buchi$ "condition", all transitions are "accepting@@state" except the self-loops $v\re{v}v$.
	This automaton recognises the paths in $G$, allowing repetition of vertices, but that visit at least two different vertices infinitely often.

	Let $k$ be the size of a minimal vertex cover of $G$. We claim that there is a "state-based" "B\"uchi" automaton with $|V|+k$ states admitting a "locally bijective morphism" to $\A_G$, and that this is optimal.
	To obtain such a "state-based" automaton, we duplicate every state $v$ that is part of a given vertex cover. Let $v_{\bullet},v_{-}$ be the two copies of this state, and set $v_{\bullet}$ to be an "accepting state". 
	Among non-duplicated states, transitions are as in $\A_G$.
	For duplicated states, we let $v_i \re{v} v_{-}$ for $i \in\{-,\bullet\}$ and $u_i\re{v}v_{\bullet}$ for $(u,v)\in E$. It is easy to chech that $\pp(v_i) = v$ defines a "locally bijective morphism".

	For the converse direction, let  $\A$ be a "state-based" "B\"uchi" automaton and $\pp\colon \A \to \A_G$ a "locally bijective morphism". For every state $v$ in $\A_G$, $\pp^{-1}(v)$ must contain a non-"accepting state", as a run ending in $v^\omega$ is rejecting in $\A_G$.
	We claim that the set of vertices such that  $\pp^{-1}(v)$ contains an "accepting state" is a vertex cover of $G$.
	Indeed, for every edge $(u,v)\in E$, a word ending in $(uv)^\omega$ is accepting in $\A_G$, therefore, either $\pp^{-1}(u)$ or $\pp^{-1}(v)$  contains an "accepting state".

	The problem is in $\NP$, as there is always such an automaton with $2|Q|$ states. For $n< 2|Q|$, it suffices to guess an automaton $\A_{\st}$ with $n$ states and a "locally bijective morphism" $\pp\colon \A_{\st} \to \A_{\tr}$. %Checking correctness of $\pp$ can be done in polynomial time.
	\end{proof}

	In our opinion, the above propositions indicate that "state-based" acceptance is often innapropriate.
    We believe that, in an ideal scenario, each state of a minimal automaton should stand for some semantic properties of the language they represent (in the case of automata over finite words, these are the residuals of the language).
	This cannot be the case for "state-based" "$\oo$-automata", as some states must be allocated to encode parts of the acceptance condition.

\section{Minimisation and transformations of automata}\label{sec:automata}
In this section we study three problems relating to $\omega$-automata: minimisation, conversion of "acceptance condition" and determinisation. 
We discuss how the use of "transition-based" or "state-based" acceptance can critically affect these problems. 

\subsection{Minimisation of coB\"uchi automata}
\AP The ""minimisation problem"" asks, given an automaton and a number $n$, whether there is an equivalent automaton with at most $n$ states.
This problem admits different variants, depending on the class of automata that constitutes the search space (here we  assume that this class is the same for the input and output automata).

In 2010, Schewe showed that the "minimisation problem" is $\NP$-hard for most types of deterministic "state-based" $\omega$-automata, including "B\"uchi", "coB\"uchi" or "parity"~\cite{Schewe10MinimisingNPComplete}.
It came as a surprise when Abu Radi and Kupferman showed that "history-deterministic" "coB\"uchi" automata can be minimised in polynomial time~\cite{AK22MinimizingGFG} (conference version from 2019~\cite{AK19Min}).
Soon after, Schewe showed that the same problem is $\NP$-hard for "state-based" automata.\footnote{Note that the critical difference lies in the output class, as we can convert the input from "state-based" to "transition-based" in polynomial time.}

\AP An automaton is ""history-deterministic"" (abbreviated HD) if there is a resolver $\sigma\colon \Sigma^*\times \Sigma \to \Delta$, such that for every word $w$ accepted by the automaton, the run over $w$ built following the transitions given by $\sigma$ is accepting.
"History-deterministic" "coB\"uchi" automata are as expressive as deterministic ones, but they can be exponentially more succinct~\cite{KS15DeterminisationGFG}.

\begin{proposition}[\cite{AK22MinimizingGFG},\cite{Schewe20MinimisingGFG}]
	"History-deterministic" "transition-based" "coB\"uchi" automata can be minimised in polynomial time.

	The "minimisation problem" for "history-deterministic" "state-based" "coB\"uchi" automata is $\NPc$.
\end{proposition}

The work of Abu Radi and Kupferman provided the basis of many subsequent results, including new representations for $\omega$-regular languages~\cite{EhlersSchewe22NaturalColors,Ehlers25Rerailing}, minimisation of "HD" "generalised coB\"uchi" automata~\cite{CIKM025}, passive learning of "HD" "coB\"uchi" automata~\cite{LW25congruences} and characterisations of "positional languages"~\cite{CO24Positional}.
The "transition-based" assumption is essential to all these works.

Schewe's proof of $\NP$-hardness of the "minimisation" of deterministic "state-based" "B\"uchi" automata~\cite{Schewe10MinimisingNPComplete} strongly relies on putting the acceptance over states.
In fact, as we have seen in Proposition~\ref{prop:from-trans-to-states-optimal-NP}, what this reduction shows is that finding a minimal "state-based" automaton that simulates a "transition-based" one is $\NP$-hard.
It was not until 2025 that the "minimisation" of deterministic "transition-based" "B\"uchi" and "coB\"uchi" automata was shown to be $\NP$-hard, requiring a highly technical proof~\cite{AE25Minimisation}.

\subsection{Translation from Muller to parity}
The complexity of the "acceptance condition" used by an automaton may greatly affect the computational cost of dealing with these automata.
Namely, many problems are $\PSPACE$-hard for "Muller" automata~\cite{Dawar2005ComplexityBounds}, but become tractable for "parity" "automata"~\cite{CJKLS22,Boker18WhyTypes}.
%Namely, solving "Muller" "games" is $\PSPACE$-complete~\cite{Dawar2005ComplexityBounds}, while highly efficient algorithms exist for solving "parity" "games" in practice~\cite{Dijk18Oink} (in theory, solving parity games is in $\NP \cap \coNP$ and can be done in quasipolynomial time~\cite{CJKLS22}). 
Therefore, an important task is to simplify the "acceptance condition" of a given automaton.
In practice, this usually takes the following form: given an "automaton" using a "Muller" "condition", build an equivalent automaton using a "parity" condition. 

\AP The "parity" and "Muller" conditions are defined as follows:
\begin{itemize}
    \item $\intro*\parity(d) = \{w\in \{1,\dots,d\}^\omega \mid \liminf w \text{ is even}\}$.
	\item $\intro*\Muller(\F) = \{w\in C^\omega \mid \infOften(w)\in \F\}$, for $\F\subseteq \P(C)$ a family of subsets and $\intro*\infOften(w)$ the set of colours that appear infinitely often in $w$.
\end{itemize}

%In 2021, Casares, Colcombet and Fijalkow~\cite{CCF21Optimal} introduced 
Recently, an optimal transformation  has been introduced -- based on a structure called the \emph{Alternating Cycle Decomposition} (ACD) -- transforming a "Muller" automaton $\A$ into a parity one~\cite{CCFL24FromMtoP}.
Formally, it produces a "transition-based" "parity" automaton that admits a "locally bijective morphism" to $\A$ and with a minimal number of states among "parity" automata admiting such a morphism.
This transformation can be performed in polynomial time provided that the ACD can be computed efficiently; this is the case for example if the "acceptance condition" of $\A$ is "generalised B\"uchi", defined as follows: \AP
\begin{itemize}
	\item $\intro*\genBuchi = \{w\in \P(C)^\omega \mid \bigcup\limits_{A\in \infOften(w)}A = C \}$.
\end{itemize}

\begin{proposition}[{Follows from~\cite[Thm. 5.35]{CCFL24FromMtoP}}]
	Given a "generalised B\"uchi" automaton~$\A$, we can build in polynomial time a "transition-based" "B\"uchi" automaton admiting a "locally bijective morphism" to $\A$ that has a minimal number of states among "B\"uchi" automata admitting "locally bijective morphisms" to $\A$.
\end{proposition}

However, the optimality result of the ACD-transformation strongly relies on the use of "transition-based" acceptance in the output automaton, as the previous problem becomes $\NP$-hard for "state-based" automata.

\begin{proposition}\label{prop:transformations-NP}
	The following problem is $\NPc$:
	%\vspace{-5mm}
	\begin{center}
		%\hspace{-3mm}
	\begin{tabular}{rl}
				{Input:} & A "state-based" "generalised B\"uchi" automaton $\A$ and an integer $n$.\\[1mm]   
				{Question:} & Is there a "state-based" "B\"uchi" automaton with $n$ states admitting\\
				&  a "locally bijective morphism" to $\A$?  
		\end{tabular}
	\end{center}
\end{proposition}
\begin{proof}
	We can use the same reduction as in the proof of Proposition~\ref{prop:from-trans-to-states-optimal-NP} (which in turn comes from~\cite{Schewe10MinimisingNPComplete}).
	Indeed, we can replace the "transition-based" "B\"uchi" condition of the automaton $\A_G$ by a "state-based" "generalised B\"uchi" "condition". 
\end{proof}

\subsection{Determinisation of B\"uchi automata}

The determinisation of "B\"uchi" automata is a fundamental problem in the theory of $\oo$-automata, studied since the introduction of the model~\cite{Buchi1962decision}. The first asymptotically optimal determinisation construction is due to Safra~\cite{Safra1988onthecomplexity}, which transforms a "B\"uchi" automaton into a deterministic Rabin one.
In 1999, Redziejowski proposed a variant for building a "transition-based" automaton from a given $\omega$-regular expression~\cite{Redziejowski99}.
Later on, Piterman~\cite{Piterman2006fromNDBuchi} and Schewe~\cite{Schewe2009tighter} further improved Safra's construction, reducing the number of states of the final automaton (see also~\cite{Redziejowski12}).
\AP Schewe's construction transforms a "B\"uchi" automaton of size $n$ into a deterministic Rabin "automaton" of size at most $\intro*\sizeDet(n)$, which is naturally equipped with a "transition-based" "acceptance condition" (with $\sizeDet(n) = o((1.65n)^n)$).
In 2009, Colcombet and Zdanowski~\cite{Colcombetz2009tight}  showed that the Piterman-Schewe construction is tight (up to $0$ states!) as we precise now.

\begin{proposition}[\cite{Colcombetz2009tight}]
	 There exists a family of "B\"uchi" automata $\A_n$ with $n$ states, such that a minimal "transition-based" deterministic Rabin automaton equivalent to $\A_n$ has $\sizeDet(n)$ states.
\end{proposition}

We can obtain a "state-based" automaton by augmenting the number of states, but doing so we no longer have a matching lower bound. 
No such tight bounds are known for the determinisation of "B\"uchi" automata towards "state-based" "automata".

The complementation and determinisation problems for "B\"uchi" and "generalised B\"uchi" automata with "transition-based" acceptance were further studied by Varghese in his PhD Thesis~\cite{Varghese2014PhD}. In the works of Schewe and Varghese~\cite{Varghese12, Varghese14Determinising}, they point out the suitability of "transition-based acceptance" for the study of transformations of automata.

\section{Games on graphs and strategy complexity}\label{sec:games}

\AP A ""game"" is given by a directed graph $G=(V,E)$ with a partition of vertices into those controlled by a player Eve and those controlled by a player Adam, a initial vertex and a ""winning condition"" defined in the same way as the "acceptance condition" of "automata" (which can be "state-based" or "transition-based").
The players move a token in turns producing an infinite path, and Eve wins if this path belongs to the "winning condition".

An important concept with applications for the decidability of logics~\cite{BL69Strategies,Gurevich1982trees} and verification~\cite{BloemCJ2018Handbook} is that of strategy complexity: how complex is it to represent a winning strategy?
\AP The simplest kind of strategies are \emph{"positional"} ones. A strategy is ""positional@@strat"" if it can be represented by a function $\sigma\colon V\to E$: when in a vertex $v$ controlled by Eve, she plays the transition $\sigma(v)$.
\AP More generally, a strategy is said to use ""finite-memory"" if the choice at a given moment only depends on a finite amount of information from the past, or, said differently, it can be implemented by a finite automaton (we refer to~\cite[Section~1.5]{bookGames25} for formal definitions).

As already noticed by Zielonka~\cite{Zielonka1998infinite}, and as we will see next, 
%Kopczyński~\cite{Kopczynski2006Half}, 
strategy complexity is quite sensitive to the placement of the "winning condition".

\subsection{Bipositionality over infinite games}
\AP We say that a language $\win\subseteq C^\omega$ is ""positional@@lang"" if for every game with "winning condition" $\win$, if Eve has a winning strategy, she has a "positional@@strat" one.
\AP A language $\win$ is ""bipositional"" if both $\win$ and its complement are "positional@@lang", or, said differently, if both Eve and Adam can play optimally using "positional strategies".
Depending on whether we consider games with "transition-based" or "state-based" "winning condition", we will say accordingly \emph{positional over transition/state-based games}.

A celebrated result in the area is the proof of "bipositionality" of "parity" languages~\cite{EmersonJutla91Determinacy,Mostowski1984RegularEF}.
In 2006, Colcombet and Niwiński proved that these are the only "prefix-independent" "bipositional" languages over infinite game graphs~\cite{CN06}, establishing an elegant characterisation of "bipositionality". As indicated in the title of their paper, this characterisation only holds for "transition-based" games.

\begin{proposition}[\cite{CN06}]
	A "prefix-independent" language $\win\subseteq C^\oo$ is "bipositional" over "transition-based" "games" if and only if there is $d\in \NN$ and a mapping $\phi \colon C \to \{1,\dots,d\}$ such that $w\in \win$ if and only if $\phi(w)\in\parity(d)$.
\end{proposition}

\begin{proposition}[{\cite[Section~6]{Zielonka1998infinite}}]
	There is a "prefix-independent" language that is  "bipositional" over  totally-coloured "state-based" games, but is not equivalent to $\parity(d)$ for any~$d$.
\end{proposition}
\begin{proof}[Proof sketch]
	An example of such a language is 
	\[ \win = \{w\in\{a,b\}^\omega \mid \text{ both $a$ and $b$ appear infinitely often in }w\}.\]
	Intuitively, if Eve is in a vertex coloured $a$, she can follow a strategy leading to a vertex coloured $b$ in a positional way (and vice-versa).

	From Adam's point of view, if he can win, there are some vertices from which he can force to never produce `$a$' or force to never produce `$b$' (and this can be done positionally). Removing those vertices, we define a "positional strategy" recursively.
	(Note that this can also be done for "transition-based" games, in fact, from Adam's point of view, $\win$ is a Rabin condition, which are "positional".)
\end{proof}

The characterisation of "bipositionality" was generalised to all (not necessarily "prefix-independent") languages in~\cite[Thm.~7.1]{CO24Positional}.
A necessary condition for "bipositionality" is that the language should be recognised by a "transition-based" deterministic "parity" automaton with one state per residual of the language.
This property is very sensitive to the placement of the "acceptance condition", if suffices to consider the language $\Buchi$ that cannot be recognised by a "state-based" automaton with a single state.
The next example shows another version of this.

\begin{example}\label{ex:automaton-residuals}
	Consider the language
	\[ L = \{ w\in \{a,b\}^\omega \mid \text{ if letter `$a$' occurs in $w$ then it appears infinitely often}\}. \]
	This language has two residuals: $\ee^{-1}L$ and $a^{-1}L$. It can be recognised by a "transition-based" "parity" automaton (even a "B\"uchi" automaton) with two states, as shown in Figure~\ref{fig:buchi-bipositional}.
	One can check that it also satisfies the other conditions from~\cite[Thm.~7.1]{CO24Positional}, so it is "bipositional". 
	However, it is not possible to recognise $L$ with a "state-based" "parity" automaton with only $2$ states.
\end{example}

\begin{figure}[h]	
	\centering
		\includegraphics[scale=1.1]{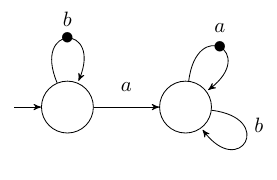}
	\caption{A "B\"uchi" automaton recognising the "bipositional" language of words that either contain no $a$, or infinitely many $a$'s.
	This automaton has one state per residual of the language.
	A "state-based" parity automaton recognising this language must have at least $3$ states.}
	\label{fig:buchi-bipositional}
\end{figure}

\subsection{Positionality via monotone graphs}
Recently, Ohlmann characterised "positionality" by means of \emph{"monotone" universal graphs}~\cite{Ohlmann23Univ}.
Not only this characterisation concerns "positionality" over "transition-based" games, but the main notion of "monotone graph" radically uses the colouring on transitions.
\AP An ordered edge-coloured graph is ""monotone"" if whenever $v\re{a}u$, $v\leq v'$ and $u'\leq u$, then the edge $v'\re{a}u'$ also appears in the graph. 
Such kind of properties can only be naturally phrased in edge-coloured graphs.

Universal "monotone" graphs have been used to study the algorithmic complexity of solving different types of games on graphs, such as parity and mean-payoff~\cite{CFGO22Universal}, and the above characterisation has been generalised to the memory of languages~\cite{CO25Universal}.

\subsection{The memory of $\omega$-regular languages}
\AP The ""memory"" of a language $\win$ is the minimal $m\in \NN$ such that in any game with objective $\win$, if Eve has a winning strategy, she has one implemented by an automaton with at most $m$ states.
A result with major implications in logic is the fact that $\oo$-regular languages have finite-memory~\cite{BL69Strategies,Gurevich1982trees}.

\AP Recently, Casares and Ohlmann gave an effective way of computing the "memory" of  $\omega$-regular languages~\cite{CO25memory}, based on a characterisation using the notion of \emph{"$\ee$-completable" "parity" automata}.
The definition of this notion is rooted in the use of "transition-based" acceptance: 
\AP A "parity" automaton is ""$\ee$-completable"" if for every pair of states $q,q'$ and even colour $x$ of the parity condition, we can either add a transition $q\re{\ee:x} q'$ or a transition $q'\re{\ee:x+1}q$ without modifying the language recognised by the automaton.

In 2023, Bouyer, Randour and Vandenhove showed that $\oo$-regular languages are exactly those that are arena-independent finite-memory determined (that is, both Eve and Adam admit finite automata implementing strategies in every game with "winning condition" $\win$)~\cite[Thm.~7]{BRV23OmegaRegMemory}.
The use of "transition-based" acceptance is key for the construction of a "parity" automaton recognising a language with the above property~\cite[Section~5]{BRV23OmegaRegMemory}.

In 2021, Casares showed that the smallest automata that can be used for implementing winning strategies in every game using a given "Muller language" $\Muller(\F)$ are exactly deterministic  Rabin automata recognising  $\Muller(\F)$~\cite[Thm.~27]{Casares2021Chromatic}.
In a related work, Casares, Colcombet and Lehtinen showed that the "memory" of $\Muller(\F)$ coincides with the number of states of a minimal "history-deterministic" Rabin automaton recognising  this language~\cite[Thm.~5]{CCL22SizeGFG}.
Both results only apply to "transition-based" Rabin automata.

\section{What about finite words?}\label{sec:finite-words}
\AP In light of the results above, one naturally wonders whether a shift to "transition-based acceptance" would also be beneficial for automata on finite words (""DFAs"" in the following).
Classical finite automata have a robust mathematical theory -- notably, every regular language admits a canonical minimal "DFA" -- and "state-based" acceptance is the undisputed preferred option for them.
However, transition-based variants have been considered recently in works about synthesis of $\LTL$ over finite traces~\cite{SXLGP20,XLZSPV21,XLHXLPSV24,DZPGV25} and about translations of regular expressions over valuations of atomic propositions~\cite{MRD24}.

%To the best of our knowledge, the first appearance of "transition-based" "DFAs" was in 2021, in a work by Xiao, Li, Zhu, Shi, Pu and Vardi about synthesis of $\LTL$ over finite traces~\cite{XLZSPV21}.

 \paragraph{Definition of acceptance.} One option to define the acceptance of transition-based finite automata is simply to specify a set of final transitions: a run is accepting if its last transition belongs to this set.\footnote{In this case, we should also specify whether the empty word is accepted.}
More generally, following the definition of "$\omega$-automata" used in this document, we define the \AP ""acceptance condition of a transition-based DFA"" as a language $\acc\subseteq \Delta^*$: a run is accepting if it belongs to $\acc$.
If $\acc$ is a regular language, such an automaton accepts a regular language (we can convert it into a classical "DFA" by a product construction).
Using a colouring function $\gamma\colon \Delta \to \{-, \circledcirc\}$ as in the introduction, we can recast acceptance by final transitions as automata using the following condition:
\[ \intro*\accLast = \{ w\in \{-, \circledcirc\}^* \mid \text{ the last letter of } w \text{ is } \circledcirc \}. \]

\paragraph{The role of prefix-independence and the empty word.}
When using the above general model of "transition-based DFAs" we encounter one inconvenience: the language recognised starting from a given state $q$ may be ill-defined, since the set of runs accepted from $q$ depends on the particular path that led to $q$ from the initial state.
Independence from the past of the run is a key property, %in the theory of automata, 
notably for defining a minimal DFA, where each state corresponds to a left-quotient of the language.

\AP This problem would not arise if the acceptance condition $\acc$ was ""prefix-independent"", that is, if for all sequences of transitions $u_0$ and $u$: 
\begin{equation*}
	\label{eq:consistency}
	u_0u \in \acc  \iff u \in \acc.%\tag{*}
\end{equation*}

However, the only "prefix-independent" languages of finite words are the empty and the full language, which cannot be used to recognise non-trivial languages. Indeed, if $\acc$ is "prefix-independent", then $u\in \acc \iff \ee\in \acc$ for all $u\in \Sigma^*$. 
%Therefore, it is not possible to obtain a transition-based model recognising non-trivial languages of finite words and with the consistency property~\eqref{eq:consistency}. 

Nevertheless, the language $\accLast$ is almost "prefix-independent", as it satisfies:
\begin{equation*}
	\label{eq:eps-prefInd}
	\text{for all } u_0 \tand u\neq \ee, \; \; u_0u\in \accLast \iff u\in \accLast \,.%\tag{*}
\end{equation*}

This property makes $\accLast$ well-suited for state-based acceptance, as the acceptance of $\ee$ can be encoded in a state, obtaining a definition of acceptance that is agnostic to the way we reach a given state.

\paragraph{Minimal transition-based DFA.} It is well-known that the minimal state-based "DFA" of a regular language $L\subseteq \Sigma^*$ is given by the equivalence classes of the Myhill-Nerode congruence.
In order to fit the "transition-based" setting, we can coarsen this relation, disregarding separations by the empty word:
\[u\intro*\eqeps v  \;\; \iffdef \;
\; \text{for all } w\neq \epsilon, \; uw\in L \iff vw\in L.\]

The next lemma is an easy check.

\begin{lemma}\label{lem:eqres-congruence}
	The relation $\eqeps$ is an equivalence relation over $\Sigma^*$. Moreover, if $u\eqeps v$, then $ua\eqeps va$ for all $a\in \Sigma$.
\end{lemma}

In the following, by a "transition-based DFA" we mean one with acceptance by final transitions, that is, using the acceptance condition $\accLast$.

\begin{proposition}\label{prop:minimal-TDFA}
	Every regular language of finite words has a unique minimal "transition-based DFA", which has one state per equivalence class of $\eqeps$. 
\end{proposition}
\begin{proof}
	Let $L\subseteq \Sigma^*$ be a regular language.
	Consider the "DFA" $\A_{\min}$ having as states the $\eqeps$-classes of $L$, with $[\epsilon]$ the initial state, and transitions $[u] \re{a} [ua]$, where accepting transitions are those with $ua\in L$. 
	Moreover, we need to specify whether $\ee\in L$; in the positive case, we let the initial transition of the automaton be accepting.
	This automaton is well-defined and recognises $L$ thanks to Lemma~\ref{lem:eqres-congruence}. 

	Let $\A$ be a "transition-based DFA" recognising $L$. For a state $q$ in $\A$, let 
	\[ L_\ee^{\A}(q) = \{ w\in \Sigma^+ \mid \text{ the run over } w \text{ from $q$ is accepting}\}.\]
	It holds that, if $u$ labels a path from the initial state to $q$, then $L_\ee^{\A}(q)= L_\ee^{\A_{\min}}([u])$. 
	Moreover, $L_\ee^{\A_{\min}}(u) \neq L_\ee^{\A_{\min}}([v])$ if $u \not\eqeps v$.
	Therefore, $\A_{\min}$ has at most as many states as $\A$, and in case of equality, they are isomorphic.
\end{proof}

Proposition~\ref{prop:minimal-TDFA} implies that "transition-based DFAs" are not larger than state-based ones.
Moreover, they can be strictly smaller, as shown by the following example and Corollay~\ref{cor:size} (see also~\cite[Figs.~2-4]{MRD24}).
% Figure~\ref{fig:finite-automata} and proved in the next proposition.

\begin{example}
	Let $\Sigma = \{a,b\}$ and consider the language of words that either have even lenght and end by `$b$', or have odd length and end by `$a$'.
	A minimal state-based DFA for this language, with $4$ states, is given on the left of Figure~\ref{fig:finite-automata}.
	Note that the states $q_a$ and $q_b$ are not equivalent, as only one of them is accepting.
	However, $L_\ee(q_a) = L_\ee(q_b)$ (idem for $p_a$ and $p_b$).
	Therefore, we can merge these states, obtaining a transition-based DFA with only $2$ states.
	%\[L = \{w\in \Sigma^+ \mid w \text{ is of even lenght if and only if it ends by } `b'\}\]
\end{example}

\begin{figure}[ht]
	\centering
	\hspace{5mm}
	\begin{minipage}[c]{0.45\textwidth} 
		\includegraphics[scale=1.1]{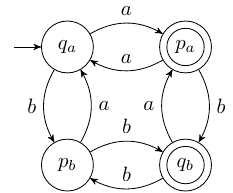}
	\end{minipage} 
	\hspace{-5mm}
	\begin{minipage}[c]{0.45\textwidth} 
		\includegraphics[scale=1.1]{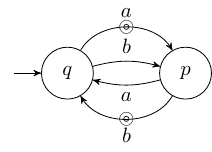}
	\end{minipage} 
	\caption{Two automata recognising the language $L = \{w\in \{a,b\}^+ \mid w \text{ is of even lenght if and only if it ends by  `$b$'}\}$. % even length and ending by `$a$'.
	The automaton on the left is the minimal state-based "DFA" of $L$, and the automaton on the right is its minimal "transition-based DFA".}
	\label{fig:finite-automata}
\end{figure}

Generalising the previous example  and using Proposition~\ref{prop:minimal-TDFA}, we obtain:

\begin{corollary}\label{cor:size}
	Every "transition-based DFA" admits an equivalent state-based "DFA" with at most twice as many states.
	This bound is tight: there is a family of languages for which a minimal state-based "DFA" has twice as many states as a minimal "transition-based DFA".
\end{corollary}
\begin{proof}
	For the first claim, it suffices to note that the index of the classical Myhill-Nerode congruence is at most twice the index of $\eqeps$.
	
	For the second claim, let $\Sigma = \{a,b\}$ and consider the language:
	\[ L_n = \{ w\in \Sigma^+ \mid w \text{ ends by `$b$' if and only if } \; |w| \equiv 0 \bmod n\}. \]
	The congruence $\eqeps$ has $n$ classes, corresponding to the remainder of $|w|$ modulo $n$. 
	The Myhill-Nerode congruence has $2n$ classes, as words ending with a different letter are not equivalent. 
	%State-based and transition-based automata for $L_2$ are depicted in Figure~\ref{fig:finite-automata}. 
\end{proof}

We note that such a gap does not appear for non-deterministic automata.
Indeed, every transition-based NFA can be converted into an equivalent state-based NFA with only one more state. It suffices to add a sink state, which will be the only accepting state, and duplicate all accepting transitions redirecting one copy towards this sink. (This operation increases the number of transitions though.)

\paragraph{Where does this leave us?}
As we have seen, "transition-based DFAs" can be smaller than classical state-based ones.
%This can come in handy in practical applications, especially in situations where the acceptance of the empty word -- one technical inconvenience for the "transition-based" definition -- is irrelevant.
Moreover, most standard constructions adapt to the transition-based setting without problem (determinisation, product construction, removal of $\ee$-transitions, etc). Some of them, such as the conversion of regular expressions, may even benefit from the use of transition-based acceptance~\cite{MRD24}.
Transition-based DFAs can be of particular interest when used as an intermediate step for the construction of $\omega$-automata~\cite{MRD24}, or when the acceptance of the empty word is irrelevant, as in $\LTLf$ semantics.
%"DFAs" also appear well-suited for on-the-fly constructions~\cite{XLZSPV21,XLZSLPV24,DZPGV25} and 

However, there are signs pointing towards the canonicity of state-based acceptance for DFAs.
Notably, the syntactic monoid of a language equals the transition monoid of its minimal state-based DFA~\cite[Prop.~4.28]{Pin25MPRI}.
It is unclear to us what would be the correct way to recover the syntactic monoid from a "transition-based DFA".
Other questions regarding the "transition-based" model remain open.
For instance, are there acceptance conditions other than $\accLast$ that lead to unique minimal DFAs for all regular languages?

\section{Outlook: Why all these differences?}\label{sec:conclusion}
%\subsection{Why is transition-based acceptance better behaved?}
We have seen various situations where "transition-based" acceptance
is more advantageous, both for practical and theoretical reasons. The following question
arises naturally:
What are the fundamental differences between "state-based" and "transition-based" models that lead to such contrasting properties?

\subparagraph{Composition of transitions.}
A basic operation at the heart of many reasonings in automata theory is \emph{composition of transitions}.
If an automaton contains transitions $p \re{a} q$ and $q\re{b} r$, one can go from $p$ to $r$ by reading $ab$, and any ``effect'' of this path should be the result of concatenating the effects of these two transitions.
That is, a suitable automata model should allow to add the transition $p\re{ab}r$. For automata over infinite words, the acceptance of the automaton obtained by adding this transition can only be defined in a sensible way by using a "transition-based" condition.

This composition operation is key for the celebrated connection between automata and algebra.
The suitability of transition-based models for algebraic approaches is explicitly mentionned in Michel's work introducing transition-based $\omega$-automata~\cite[Section~II]{Michel1984}:

\begin{quote}
	{\small \textit{Using unstable graphs, instead of a set of nodes that must be traversed infinitely often, is better suited to the algebraic operations we will define~[...]}}~\footnote{In French in the original: {\small \textit{L’utilisation de graphes instables au lieu d’un ensemble de nœuds dans lequel on doit passer infiniment souvent se prête mieux aux opérations algébriques que nous définirons~[...]}.}}
\end{quote}

Similarly, one of LeSaëc's motivations for the use of transition-based automata was to obtain an algebraic proof of McNaughton's theorem for infinite words~\cite{SPW91Semigroups}.
The "Muller" automaton obtained from a given semigroup is naturally "transition-based", see~\cite[page~18]{SPW91Semigroups} and~\cite[Section~6]{Colcombet11Green}.

As mentioned in Section~\ref{sec:games}, composition of transitions is also essential in the fruitful approach for solving and analysing infinite duration games based on universal graphs, which relies on the notions of "monotonicity", "$\ee$-completion" and the technique of saturation (for the latter, see~\cite[Section~4]{CF18UniversalGraphs},~\cite[Section~4.1]{CFGO22Universal} or~\cite[Section~3.3]{Ohlmann23Univ}).

We note, however, that in the case of finite words, this does not provide strong evidence in favour of transition-based acceptance.
Indeed, state-based DFAs also allow for composition of transitions, as the acceptance of a run is only determined by its final destination. 

\subparagraph{Paths in graphs.} 
As explained in the introduction, an "acceptance condition" is a representation of a subset of paths in an automaton.
A path in a graph is commonly defined as a sequence of edges. 
In fact, a sequence of vertices does not completely determine a path, as different paths may share the same sequence of vertices.
This is the main reason why "transition-based" automata are more succinct than "state-based" ones.

\subsection*{Final thoughts}
The collection of results presented in this survey indicates that, despite the fact that the size of "state-based" and "transition-based" automata only differ by a linear factor, "transition-based" models are easier to manipulate and have a nicer theory.
We therefore advocate adopting "transition-based acceptance" as the default model for $\oo$-automata.

We expect that the use of "transition-based" acceptance will ease the finding of automata-based characterisation of classes of languages. This has already been the case, for example, in the characterisation of "positional" $\oo$-regular languages based on "parity" automata with a particular structure~\cite[Thm.~3.1]{CO24Positional}.

In the same spirit, it appears that the use of "transition-based" models will be required for obtaining canonical models of automata over infinite words or trees. Steps in this direction have already been made~\cite{EhlersSchewe22NaturalColors,Ehlers25Rerailing,LW25congruences}, building on the description of canonical "history-deterministic" "coB\"uchi" automata by Abu Radi and Kupferman~\cite{AK22MinimizingGFG}.

\subparagraph{ Acknowledgements.}
 I warmly thank Thomas Colcombet for many discussions on the benefits of transition-based acceptance, Alexandre Duret-Lutz for valuable comments on automata over finite words and for sharing many historical references, Géraud Sénizergues for pointing me to the works of Bertrand Le Saëc and 
Pierre Ohlmann for helpful feedback on a draft of this paper.

%\bibliographystyle{alpha}
%\bibliographystyle{plainurl}
%\bibliography{references.bib}

\printbibliography

% \newpage
% \appendix
% \section{Missing NP-completeness proofs}\label{appendix:NP}
% %\input{appendix-NP.tex}

\end{document}